\documentclass{sig-alternate}

\usepackage{amssymb}
\usepackage{dsfont}
\usepackage{amsmath}
\usepackage[noend]{algorithmic}
\usepackage[ruled,vlined]{algorithm2e}
\usepackage{multirow}

\newtheorem{theorem}{Theorem}
\newdef{definition}{Definition}
\newtheorem{proposition}{Proposition}
\begin{document}
%

\title{Mining Representative Unsubstituted Graph Patterns Using Prior Similarity Matrix\titlenote{An implementation of the proposed approach and the datasets used in the experiments are freely available on the first authors personal web page: http://fc.isima.fr/$\sim$dhifli/unsubpatt/ or by email request to any one of the authors}\titlenote{This paper is the full version of a preliminary work accepted as abstract in MLCB'12 (NIPS workshop)}}
%
%
%
%
%

\numberofauthors{3} 
%
\author{
%
%
\alignauthor
Wajdi Dhifli\\
      \affaddr{Clermont University, Blaise Pascal University, LIMOS, BP 10448, F-63000 Clermont-Ferrand, France}\\
       \affaddr{CNRS, UMR 6158, LIMOS, F-63173 Aubi\`ere, France}\\
       \email{dhifli@isima.fr}
\alignauthor
Rabie Saidi\\
       \affaddr{European Bioinformatics Institute}\\
       \affaddr{Hinxton, Cambridge, CB10 1SD, United Kingdom}\\
       \email{rsaidi@ebi.ac.uk}
\alignauthor 
Engelbert Mephu Nguifo\\
      \affaddr{Clermont University, Blaise Pascal University, LIMOS, BP 10448, F-63000 Clermont-Ferrand, France}\\
       \affaddr{CNRS, UMR 6158, LIMOS, F-63173 Aubi\`ere, France}\\
       \email{mephu@isima.fr}
}

\maketitle
\begin{abstract}
One of the most powerful techniques to study protein structures is to look for recurrent fragments (also called substructures or spatial motifs), then use them as patterns to characterize the proteins under study. An emergent trend consists in parsing proteins three-dimensional (3D) structures into graphs of amino acids. Hence, the search of recurrent spatial motifs is formulated as a process of frequent subgraph discovery where each subgraph represents a spatial motif. In this scope, several efficient approaches for frequent subgraph discovery have been proposed in the literature. However, the set of discovered frequent subgraphs is too large to be efficiently analyzed and explored in any further process. In this paper, we propose a novel pattern selection approach that shrinks the large number of discovered frequent subgraphs by selecting the representative ones. Existing pattern selection approaches do not exploit the domain knowledge. Yet, in our approach we incorporate the evolutionary information of amino acids defined in the substitution matrices in order to select the representative subgraphs. We show the effectiveness of our approach on a number of real datasets. The results issued from our experiments show that our approach is able to considerably decrease the number of motifs while enhancing their interestingness.
\end{abstract}


\category{E.4}{Coding and information theory}[Database Applications - Data Mining]
\terms{Algorithms, Experimentation}



\section{Introduction}
Studying protein structures can reveal relevant structural and functional information which may not be derived from protein sequences alone. During recent years, various methods that study protein structures have been elaborated based on diverse types of descriptor such as profiles \cite{Niklas_2004}, spatial motifs \cite{Kleywegt_1999,Sun_2012} and others. Yet, the exponential growth of online databases such as the Protein Data Bank (PDB) \cite{Berman_2000}, CATH \cite{Alison_2O11}, SCOP \cite{Andreeva_2008} and others, arises an urgent need for more accurate methods that will help to better understand the studied phenomenons such as protein evolution, functions, etc.

In this scope, proteins have recently been interpreted as graphs of amino acids and studied based on graph theory concepts \cite{Huan_2004}. This representation enables the use of graph mining techniques to study protein structures in a graph perspective. In fact, in graph mining, any problem or object under consideration is represented in the form of nodes and edges and studied based on graph theory concepts \cite{Bartoli_2007,Zaki_2009,Jin_2009,Cheng_2010}. One of the powerful and current trends in graph mining is frequent subgraph discovery. It aims to discover subgraphs that frequently occur in a graph dataset and use them as patterns to describe the data. These patterns are lately analyzed by domain experts to reveal interesting information hidden in the original graphs, such as discovering pathways in metabolic networks \cite{Faust_2010}, identifying residues that play the role of hubs in the protein and stabilize its structure \cite{Vallabhajosyula_2009}, etc.

The graph isomorphism test is one of the main bottlenecks of frequent subgraph mining. Yet, many efficient and scalable algorithms have been proposed in the literature and made it feasible for instance FFSM \cite{Huan_2003}, gSpan \cite{Yan_2002}, GASTON \cite{Nijssen_2004}, etc. Unfortunately, the exponential number of discovered frequent subgraphs is another serious issue that still needs more attention \cite{Woznica_2012}, since it may hinder or even make any further analysis unfeasible due to time, resources, and computational limitations. For example, in an AIDS antiviral screen dataset composed of 422 chemical compounds, there are more than 1 million frequent substructures when the minimum support is 5\%. This problem becomes even more serious with graphs of higher density such as those representing protein structures. In fact, the issues raised from the huge number of frequent subgraphs are mainly due to two factors, namely \textit{redundancy} and \textit{significance} \cite{Marisa_2010}. Redundancy in a frequent subgraph set is caused by structural and/or semantic similarity, since most discovered subgraphs differ slightly in structure and may infer similar or even the same meaning. Moreover, the significance of the discovered frequent subgraphs is only related to frequency. This yields an urgent need for efficient approaches allowing to select relevant patterns among the large set of frequent subgraphs.

In this paper, we propose a novel selection approach which selects a subset of representative patterns from a set of labeled subgraphs, we term them \textit{unsubstituted patterns}. In order to select these unsubstituted patterns and to shrink the large size of the initial set of frequent subgraphs, we exploit a specific domain knowledge, which is the substitution between amino acids represented as nodes. Though, the main contribution of this work is to define a new approach for mining a representative summary of the set of frequent subgraphs by incorporating a specific background domain knowledge which is the ability of substitution between nodes labels in the graph. In this work, we apply the proposed approach on protein structures because of the availability of substitution matrices in the literature, however, it can be considered as general framework for other applications whenever it is possible to define a matrix quantifying the possible substitution between the labels. Our approach can also be used on any type of subgraph structure such as cliques, trees and paths (sequences). In addition, it can be easily coupled with other pattern selection methods such as discrimination or orthogonality based approaches. Moreover, this approach is unsupervised and can help in various mining tasks, unlike other approaches that are supervised and dedicated to a specific task such as classification. 

The remainder of the paper is organized as follows. Section \ref{sec:related_works} discusses the recent related works in the area of pattern selection for subgraphs. In Section \ref{sec:unsub_patt_selection}, we present the background of our work and we define the preliminary concepts as well as the main algorithm of our approach. Then, Section \ref{sec:experiments} describes the characteristics of the used data and the experimental settings. Section \ref{sec:results} presents the obtained experimental results and the discussion. It is worth noting that in the rest of the paper, we use the following terms interchangeably : spatial motifs, patterns, subgraphs.

\section{Related Works}
\label{sec:related_works}
Recently, several approaches have been proposed for pattern selection in subgraph mining. In \cite{Chaoji_2008}, authors proposed ORIGAMI, an approach for both subgraph discovery and selection. First they randomly mine a sample of maximal frequent subgraphs, then straightforwardly they select an $\alpha$-orthogonal (non-redundant), $\beta$-representative subgraphs from the mined set. The LEAP algorithm proposed in \cite{Yan_2008} tries to locate patterns that individually have high discrimination power, using an objective function score that measures each pattern's significance. Another approach termed gPLS  was proposed in \cite{Saigo_2008}. It attempts to select a set of informative subgraphs in order to rapidly build a classifier. gPLS uses the mathematical concept of partial least squares regression to create latent variables allowing a better prediction. COM \cite{Jin_2009} is another subgraph selection approach which follows a process of pattern mining and classifier learning. It mines co-occurrence rules. Then, based on the co-occurrence information it assembles weak features in order to generate strong ones. In \cite{Marisa_2010}, authors proposed a feature selection approach termed CORK. To find frequent subgraphs, it uses the state-of-the-art approach gSpan. Then using a submodular quality function, it selects among them the subset of subgraphs that are most discriminative for classification. In \cite{Fei_2010}, authors designed LPGBCMP, a general model which selects clustered features by considering the structure relationship between subgraph patterns in the functional space. The selected subgraphs are used as weak classifiers (base learners) to obtain high quality classification models. To the best of our knowledge, in all existing subgraph selection approaches \cite{Hasan_2008}, the selection is usually based on structural similarity \cite{Chaoji_2008} and/or statistical measures (e.g. frequency and coverage (closed \cite{Yan_2003}, maximal \cite{Thomas_2010}), discrimination \cite{Marisa_2010}, ...). Yet, the \textit{prior} information and knowledge about the domain are often ignored. However, these prior knowledge may help building dedicated approaches that best fit the studied data.

\section{Mining Representative Unsubstituted Patterns}
\label{sec:unsub_patt_selection}

\subsection{Background}
Statistical pattern selection methods have been widely used to resolve the dimensionality problem when the number of discovered patterns is too large. However, these methods are too generic and do not consider the specificity of the domain and the used data. We believe that in many contexts, it would be important to incorporate the background knowledge about the domain in order to create approaches that best fit the considered data. In proteomics, a protein structure is composed by the folding of a set of amino acids. During evolution, amino acids can substitute each other. The scores of substitution between pairs of amino acids were quantified by biologists in the literature in the form of substitution matrices such as Blosum62 \cite{Eddy_2004}. Our approach uses the substitution information given in the substitution matrices in order to select a subset of unsubstituted patterns that summarizes the whole set of frequent subgraphs. We consider the selected patterns as the representative ones.

The main idea of our approach is based on node substitution. Since the nodes of a protein graph represent amino acids, though, using a substitution matrix, it would be possible to quantify the substitution between two given subgraphs. Starting from this idea, we define a similarity function that measures the distance between a given pair of subgraphs. Then, we preserve only one subgraph from each pair having a similarity score greater or equal to a user specified threshold such that the preserved subgraphs represent the set of representative unsubstituted patterns. An overview of the proposed approach is illustrated in Figure \ref{fig:pattern_mining} and a more detailed description is given in the following sections.

\begin{figure}
 \centering
	\includegraphics[width=0.45\textwidth]{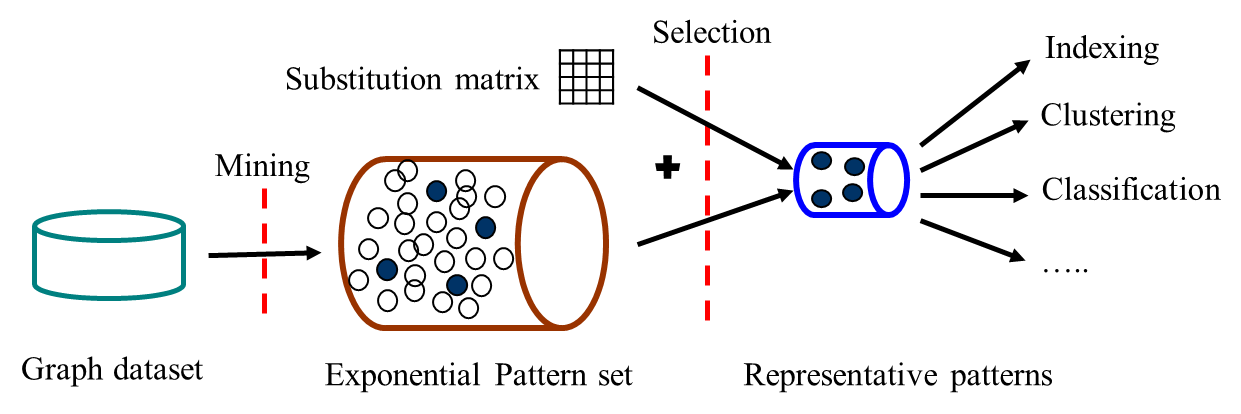}
		\caption{Unsubstituted pattern selection framework.}\label{fig:pattern_mining}
\end{figure}

The substitution between amino acids was also used in the literature but for sequential feature extraction from protein sequences in \cite{Saidi_2010}, where the authors proposed a novel feature extraction approach termed \textit{DDSM} for protein sequence classification. Their approach is restricted to protein sequences and generates every subsequence substituting another one. In other words, DDSM eliminates any pattern substituted by another one and which itself does not substitute any other one. We believe that their approach does not guarantee an optimal summarization since its output may still contain patterns that substitute each other. In addition, they do consider negative substitution scores as impossible substitutions which is biologically not true since negative scores are only expressing the less likely substitutions, which obviously does not mean that they are impossible. Moreover, DDSM is limited to protein sequences and does not handle more complex structures such as the protein tertiary structure. Our approach overcomes these shortcomings, since it can handle both protein sequences and structures (since a sequence can be seen as a line graph). In addition, it consider both the positive and negative scores of the matrix. Moreover, our approach generates a set of representative unsubstituted patterns ensuring an optimal summarization of the initial set. Besides, it is unsupervised and can be exploited in classification as well as other analysis and knowledge discovery contexts unlike DDSM which is dedicated to classification.

\subsection{Preliminaries}\label{subsec:formalization}
In this section we present the fundamental definitions and the formal problem statement. Let $\mathcal{G}$ be a dataset of graphs. Each graph $G=(V,E,L)$ of $\mathcal{G}$ is given as a collection of nodes $V$ and a collection of edges $E$. The nodes of $V$ are labeled within an alphabet $L$. We denote by $|V|$ the number of nodes (also called the graph order) and by $|E|$ the number of edges (also called the graph size). Let also $\Omega$ be the set of frequent subgraphs extracted from $\mathcal{G}$, also referred here as $patterns$.
\begin{definition}\label{def:subst_mat}(Substitution matrix)
Given an alphabet $L$, a substitution matrix $\mathcal{M}$ over $L$ is the function defined as:
\begin{equation}
\begin{array}{l|rcl}
\mathcal{M} : & L^2 & \longrightarrow & [\bot,\top] \subset \mathds{R}\\
& (l, l') & \longmapsto & x
\end{array}
\end{equation}

The higher the value of $x$ is, the likely is the substitution of $l'$ by $l$. If $x=\bot$ then the substitution is impossible, and if $x=\top$ then the substitution is certain. The values $\bot$ and $\top$ are optional and user-specified. They may appear or not in $\mathcal{M}$. The scores in $\mathcal{M}$ should respect the following two properties:
\begin{enumerate}
\item $\forall$ \textit{l} $\in$ \textit{L}, $\exists$ \textit{l'} $\in$ \textit{L} $\mid$ $\mathcal{M}(l,l')$ $\neq$ $\bot$,
\item $\forall$ \textit{l} $\in$ \textit{L}, if $\exists$ \textit{l'} $\in$ \textit{L} $\mid$ $\mathcal{M}(l,l')$ = $\top$ then $\forall$ \textit{l"} $\in$ \textit{L}$\setminus$\{\textit{l}, \textit{l'}\}, $\mathcal{M}(l,l")$ = $\bot$ and $\mathcal{M}(l',l")$ = $\bot$.
\end{enumerate}

In many real world applications, the substitution matrices may contain at the same time positive and negative scores. In the case of protein's substitution matrices, both positive and negative values represents possible substitutions. However, positive scores are given to the more likely substitutions while negative scores are given to the less likely substitutions. Though, in order to give more magnitude to higher values of $x$, $\forall$ $l$ and $l'$ $\in$ $L$:
\begin{equation}
\mathcal{M}(l,l') =  \mathrm{e}^{\mathcal{M}(l,l')}
\end{equation}
\end{definition}

\begin{definition}\label{def:patt_shape}(Structural isomorphism)
Two patterns $P=(V_P,E_P,L)$ and $P'=(V_{P'},E_{P'},L)$ are said to be structurally isomorphic (having the same shape), we note $shape(P,P')$. $shape(P,P')=true$ , iff:
\begin{enumerate}
\item [-] $P$ and $P'$ have the same order, $i.e., |V_P|=|V_{P'}|$,
\item [-] $P$ and $P'$ have the same size, $i.e., |E_P|=|E_{P'}|$,
\item [-] $\exists$ a bijective function $f : V_P \rightarrow V_{P'} \mid \forall u,v \in V_P$ if $(u, v)\in E_P$ then $(f(u), f(v))\in E_{P'}$ and inversely.
\end{enumerate}

It is worth mentioning that in the graph isomorphism problem we test whether two graphs are exactly the same by considering both structures and labels. But in this definition, we only test whether two given graphs are structurally the same, in terms of nodes and edges, without considering the labels. 
\end{definition}

\begin{definition}\label{def:elem_mut_prob}(Elementary mutation probability)
Given a node $v$  of a label $l \in L$, the elementary mutation probability, $M_{el}(v)$, measures the possibility that $v$ stay itself and \textbf{does not} mutate to any other node depending on its label $l$.
\begin{equation}\label{eq:elem_mut_prob}
    M_{el}(v)=\begin{cases}
    0, \textit{ if }  \mathcal{M}(l,l)= \mathrm{e}^{\bot} \\
    1, \textit{ if }  \mathcal{M}(l,l)= \mathrm{e}^{\top} \\
    \frac{\mathcal{M}(l,l)}{\sum_{i=1}^{|L|} \mathcal{M}(l,l_i)}, \textit{otherwise}
\end{cases}
\end{equation}

Obviously, if the substitution score in $\mathcal{M}$ between $l$ and itself is $\bot$ then it is certain that $l$ will mutate to another label $l'$ and the probability value that $v$ does not mutate should be 0. Respectively, if this substitution score is $\top$ then it is certain that $v$ will stay itself and conserve its label $l$ so the probability value must be equal to 1. Otherwise, we divide the score that $l$ mutates to itself by the sum of all the possible mutations.
\end{definition}

\begin{definition}\label{def:patt_mut_prob}(Pattern mutation probability)
Given a pattern $P=(V_P,E_P,L) \in \Omega$, the pattern mutation probability, $M_{patt}(P)$, measures the possibility that $P$ mutates to any other pattern having the same order.
\begin{equation}\label{eq:Pi}
	M_{patt}(P) = 1 - \prod_{i=1}^{\mid V_P\mid} M_{el}(P[i])
\end{equation}

where $\prod_{i=1}^{|V_P|} M_{el}(P[i])$ represents the probability that the pattern $P$ does not mutate to any other pattern $i.e.$ $P$ stays itself.
\end{definition}

\begin{definition}\label{def:elem_subst_prob}(Elementary substitution probability)
Given two nodes $v$ and $v'$ having correspondingly the labels $l, l' \in L$, the elementary substitution probability, $S_{el}(v,v')$, measures the possibility that $v$ substitutes $v'$.
\begin{equation}\label{eq:EP}
	S_{el}(v,v') = \frac{\mathcal{M}(l,l')}{\mathcal{M}(l,l)}
\end{equation}

It is worth noting that $S_{el}$ is not bijective $i.e.$ $S_{el}(v,v')$ $\neq$ $S_{el}(v',v)$.
\end{definition}

\begin{definition}\label{def:patt_subst_prob}(Pattern substitution score)
Given two patterns $P=(V_P,E_P,L)$ and $P'=(V_{P'},E_{P'},L)$ having the same shape, we denote by $S_{patt}(P,P')$ the substitution score of $P'$ by $P$. In other words, it measures the possibility that $P$ mutates to $P'$ by computing the sum of the elementary substitution probabilities then normalize it by the total number of nodes of $P$. Formally:
\begin{equation}\label{eq:PP}
	S_{patt}(P,P') = \frac{\sum_{i=1}^{\mid V_P\mid} S_{el}(P[i],P'[i])}{\mid V_P\mid}
\end{equation}
\end{definition}

\begin{definition}\label{def:patt_subst}(Pattern substitution)
A pattern $P$ substitutes a pattern $P'$, we note $subst(P,P',\tau)=true$, iff:
\begin{enumerate}
\item $P$ and $P'$ are structurally isomorphic ($shape(P,P') = true$),
\item $S_{patt}(P,P') \ge \tau$, where $\tau$ is a user-specified threshold such that $0\% \le \tau \le 100\%$.
\end{enumerate}
\end{definition}

\begin{definition}\label{def:unsubst_patt}(Unsubstituted pattern)
Given a threshold $\tau$ and $\Omega^*$ $\in \Omega$, a pattern $P^* \in \Omega^*$ is said to be unsubstituted iff $\nexists  P \in \Omega^*$ $\mid$ $M_{patt}(P) > M_{patt}(P^*) $ and $subst(P,P^*,\tau)=true$.
\end{definition}

\begin{proposition}[Null $M_{patt}$ case]\label{prop:null_Mpatt}
Given a pattern $P = (V_p,E_p,L) \in \Omega$, if $M_{patt}(P) = 0$ then $P$ is an unsubstituted pattern. 
\end{proposition}
\begin{proof}
The proof can be simply deduced from Definitions \ref{def:elem_mut_prob} and \ref{def:patt_mut_prob}.
\end{proof}

\begin{definition}\label{def:merge_supp}(Merge support)
Given two patterns $P$ and $P'$, if $P$ substitutes $P'$ then $P$ will represent $P'$ in the list of graphs where $P'$ occurs. Formally:
\begin{equation}\label{eq:PP}
	\forall (P, P') \mid subst(P,P',\tau)=true\textit{ }then\textit{ }D_P = D_P \cup D_{P'}
\end{equation}

where $D_P$ and $D_{P'}$ are correspondingly the occurrence set of $P$ and that of $P'$.
\end{definition}

\subsection{Algorithm}\label{subsec:algorithm}
Given a set of patterns $\Omega$ and a substitution matrix $\mathcal{M}$, we propose \textsc{UnSubPatt}(see Algorithm \ref{alg:mm}), a pattern selection algorithm which enables detecting the set of unsubstituted patterns $\Omega^*$ within $\Omega$. Based on our similarity concept, all the patterns in $\Omega^*$ are dissimilar, since it does not contain any pair of patterns that are substitutable. This represents a reliable summarization of $\Omega$.

\begin{algorithm}
\label{alg:mm}
\caption{\textsc{UnSubPatt}}
\KwData{$\Omega$, $\mathcal{M}$, $\tau$, ($\bot$,$\top$) [Optional]}
\KwResult{$\Omega^*$: \{unsubstituted patterns\}}
\Begin{
$\Omega \leftarrow \{\Omega^k \leftarrow \{P \in \Omega \textit{ } \mid \textit{ } \forall (P', P") \in \Omega^k , |V_{P'}|=|V_{P"}|$ and $|E_{P'}|=|E_{P"}|\}\}$\;
\ForEach {$\Omega^k \in \Omega$}{
$\Omega^k \leftarrow sort(\Omega^k \textit{ } by \textit{ } M_{patt})$\;
\ForEach {$P \in \Omega^k$}{
\If{$M_{patt}(P)>0$}{
\ForEach {$P' \in \Omega^k \backslash P$ $|$ $M_{patt}(P') < M_{patt}(P)$}{
\If{$M_{patt}(P')>0$}{
\If{$shape(P,P')$ and $subst(P,P',\tau)$}{
$merge\_support(P, P')$\;
$remove(P', \Omega^k)$\;
}}}}}
$\Omega^* \leftarrow \Omega^* \cup \Omega^k$\;
}}
\end{algorithm}
The general process of the algorithm is described as follows: first, $\Omega$ is divided into subsets of patterns having the same number of nodes and edges. Then, each subset is sorted in a descending order by the pattern mutation probability $M_{patt}$. Each subset is browsed starting from the pattern having the highest $M_{patt}$. For each pattern, we remove all the patterns it substitutes and we merge their supports such that the preserved pattern will represent all the removed ones wherever they occurs. The remaining patterns represent the unsubstituted pattern set. Though, $\Omega^*$ can not be summarized by a subset of it but itself. Our algorithm uses Proposition \ref{prop:null_Mpatt} to avoid unnecessary computation related to patterns with $M_{patt} = 0$. They are directly considered as unsubstituted patterns, since they can not mutate to any other pattern.

\begin{theorem}
Let $\Omega$ be a set of patterns and $\Omega^*$ its subset of unsubstituted patterns based on a substitution matrix
 $\mathcal{M}$ and a threshold $\tau$, $i.e.$, \textsc{UnSubPatt} $(\Omega, \mathcal{M}, \tau, (\bot,\top))=\Omega^*$. Then : 
\begin{equation}
\textsc{UnSubPatt}(\Omega^*, \mathcal{M}, \tau, (\bot,\top) ) = \Omega^*
\end{equation}
\end{theorem}

\begin{proof}
The proof can be deduced simply from Definition \ref{def:unsubst_patt}. Given a threshold $\tau$, $\Omega^*$ can not be summarized by its subsets unless itself. Formally:
\begin{equation}
\forall P\in\Omega^*, \nexists P'\in\Omega^*|M_{patt}(P)>M_{patt}(P')\textit{ }and\textit{ }subst(P,P',\tau)
\end{equation}
\end{proof}

\subsection{Complexity}
Suppose $\Omega$ contains $n$ patterns. $\Omega$ is divided into $g$ groups, each containing patterns of order $k$. This is done in $O(n)$. Each group $\Omega^k$ is sorted in $O(|\Omega^k| * log |\Omega^k|)$. Searching for unsubstituted patterns requires browsing $\Omega^k$ $(O(|\Omega^k|))$ and for each pattern, browsing in the worst case all remaining patterns $(O(|\Omega^k|))$ to check the shape $(O(k))$ and the substitution $(O(k))$. This means that searching for unsubstituted patterns in a group $\Omega^k$ can be done in $O(|\Omega^k|^2 * k^2)$. Hence, in the worst case, the complexity of our algorithm is $O(g * m_{max}^2 * k_{max}^2)$, where $k_{max}$ is the maximum pattern order and $m_{max}$ is the number of patterns of the largest group $\Omega^k$.

\section{Experiments}\label{sec:experiments}

\subsection{Datasets} \label{sec:dataset}
In order to experimentally evaluate our approach, we use four graph datasets of protein structures, which also have been used in \cite{Yan_2008} then \cite{Fei_2010}. Each dataset consists of two classes: positive and negative. Positive samples are proteins selected from a considered protein family whereas negative samples are proteins randomly gathered from the Protein Data Bank \cite{Berman_2000}. Each protein is parsed into a graph of amino acids. Each node represents an amino acid residue and is labeled with its amino acid type. Two nodes $u$ and $v$ are linked by an edge $e(u, v) = 1$ if the euclidean distance between their two $C_\alpha$ atoms $\Delta(C_\alpha(u), C_\alpha(v))$  is below a threshold distance $\delta$. Formally:
\begin{equation}
e(u,v)=
\begin{cases}1, \emph{  }if\emph{  } \Delta(C_\alpha(u), C_\alpha(v))\leq \delta
\\
0, \emph{  }otherwise \end{cases}
\end{equation}

In the literature, many methods use this definition with usually $\delta \geq 7 $\AA\textit{ }on the argument that $C_\alpha$ atoms define the overall shape of the protein conformation \cite{Huan_2005}. In our experiments, we use $\delta = 7 $\AA.

Table \ref{tab:datasets} summarizes the characteristics of each dataset. SCOP ID, Avg.$\mid$V$\mid$, Avg.$\mid$E$\mid$, Max.$\mid$V$\mid$ and Max.$\mid$E$\mid$ correspond respectively to the id of the protein family in SCOP database \cite{Andreeva_2008}, the average number of nodes, the average number of edges, the maximal number of nodes and the maximal number of edges in each dataset.

\begin{table*}[!htpd]
	\centering
	\caption{Experimental data\label{tab:datasets}}
    \begin{tabular}{|c|c|c|c|c|c|c|c|c|}\hline
\textbf{Dataset}  & \textbf{SCOP ID} & \textbf{Family name} & \textbf{Pos.} & \textbf{Neg.} & \textbf{Avg.$\mid$V$\mid$}& \textbf{Avg.$\mid$E$\mid$} & \textbf{Max.$\mid$V$\mid$}& \textbf{Max.$\mid$E$\mid$}\\ \hline
DS1	& 52592	& G proteins & 33	& 33 &  246	&	971	&	897	&	3 544	\\ \hline
DS2	& 48942 & C1 set domains & 38	& 38 &  238	&	928	&	768	&	2 962	\\ \hline
DS3	& 56437 & C-type lectin domains & 38	& 38 & 185	&	719	&	755	&	3 016	\\ \hline
DS4	& 88854	& Kinases, catalytic subunit & 41	& 41 &  275	&	1077	&	775	&	3 016	\\ \hline
\end{tabular}

\end{table*}

\subsection{Protocol and Settings}
Generally, in a pattern selection approach two aspects are emphasized, namely the number of selected patterns and their interestingness. In order to evaluate our approach, we first use the state-of-the-art method of frequent subgraph discovery gSpan \cite{Yan_2002} to find the frequent subgraphs in each dataset with a minimum frequency threshold of 30\%. Then, we use \textsc{UnSubPatt} to select the unsubstituted patterns among them with a minimum substitution threshold $\tau$=30\%. For our approach, we use $Blosum62$ \cite{Eddy_2004} as the substitution matrix since it turns out that it performs well on detecting the majority of weak protein similarities, and it is used as the default matrix by most biological applications such as BLAST \cite{Altschul_1990}. It is worth mentioning that the choice of 30\% as minimum frequency threshold for the frequent subgraph extraction is to make the experimental evaluation feasible due to time and computational limitations. 

In order to evaluate the number of selected subgraphs, we define the selection rate as the rate of the number of unsubstituted subgraphs from the initial set of frequent subgraphs. Formally :
\begin{equation}
Selection \textit{ } rate = \frac{|\Omega^*| \ast 100}{|\Omega|}
\end{equation}

To evaluate the interestingness of the set of selected patterns, we use them as features for classification. We perform a 5-fold cross-validation classification (5 runs) on each protein-structure dataset. We encode each protein into a binary vector, denoting by "1" or "0" the presence or the absence of the feature in the considered protein. To judge the interestingness of the selected subgraphs, we use one of the most known classifier, namely the na\"ive bayes (NB) classifier, due to its simplicity and fast prediction and that its classification technique is based on a global and conditional evaluation of the input features. NB is used with the default parameters from the workbench Weka \cite{Witten_2005}.

\section{Results and Discussion}\label{sec:results}
In this section, we conduct experiments to examine the effectiveness and efficiency of \textsc{UnSubPatt} in finding the representative unsubstituted subgraphs. We test the effect of changing the substitution matrix and the substitution threshold on the results. Moreover, we study the size-based distribution of patterns and we compare the results of our approach with those of other subgraph selection methods from the literature.


\subsection{Empirical Results}
Here, we show the results of our experiments obtained in terms of number of motifs and classification results. As mentioned before, we use gSpan to extract the frequent subgraphs from each dataset with frequency $\geq$ 30\%. Then, we use \textsc{UnSubPatt} to select the unsubstituted patterns among them with a substitution threshold $\tau$=30\% and using Blosum62 as substitution matrix. At last, we perform a 5-fold cross-validation classification (5 runs) to evaluate the interestingness of each subset using the NB classifier. The obtained average results are reported in Table \ref{tab:results}.

\begin{table}[!t]
	\centering
	\caption{Number of frequent subgraphs ($\Omega$), representative unsubstituted subgraphs ($\Omega^*$) and the selection rate\label{tab:number_of_motifs}}
\begin{tabular}{|c|c|c|c|}\hline	 
\textbf{Dataset} & $\mid\Omega\mid$ & $\mid\Omega^*\mid$ & \textbf{Selection rate (\%)} \\  \hline
DS1 & 799094 & 7291 & 0.91 \\  \hline
DS2 & 258371 & 15898 & 6.15 \\  \hline
DS3 & 114792 & 14713 & 12.82 \\  \hline
DS4 & 1073393 & 9958 & 0.93 \\  \hline
\end{tabular}
\end{table}

The high number of discovered frequent subgraphs is due to their combinatorial nature (this was discussed in the introductory section). The results reported in Table \ref{tab:number_of_motifs} show that our approach decreases considerably the number of subgraphs. The selection rate shows that the number of unsubstituted patterns $\mid\Omega^*\mid$ does not exceed 13\% of the initial set of frequent subgraphs $\mid\Omega\mid$ with DS3 and even reaches less than 1\%  with DS1 and DS4. This proves that exploiting the domain knowledge, which in this case consists in the substitution matrix, enables emphasizing information that can possibly be ignored when using exiting subgraph selection approaches. 


\begin{table*}[!t]
	\centering
	\caption{Accuracy, precision, recall (sensitivity), F-score and AUC of the classification of each dataset using NB coupled with frequent subgraphs (FSg) then representative unsubstituted subgraphs (UnSubPatt)\label{tab:results}}

\begin{tabular}{|c|c|c|c|c|c|c|c|c|c|c|}\hline
\multirow{2}{*}{\textbf{Dataset}} & \multicolumn{2}{|c}{\textbf{Accuracy}} & \multicolumn{2}{|c}{\textbf{Precision}} & \multicolumn{2}{|c}{\textbf{Recall}} & \multicolumn{2}{|c}{\textbf{F-score}} & \multicolumn{2}{|c|}{\textbf{AUC}} \\ 
\cline{2-11}
 & \textbf{FSg} & \textbf{UnSubPatt} & \textbf{FSg} & \textbf{UnSubPatt} & \textbf{FSg} & \textbf{UnSubPatt} & \textbf{FSg} & \textbf{UnSubPatt} & \textbf{FSg} & \textbf{UnSubPatt} \\ \hline
DS1 & 0.62 & 0.78 & 0.61 & 0.69 & 0.70 & 0.90 & 0.64 & 0.78 & 0.64 & 0.78 \\ \hline
DS2 & 0.80 & 0.90 & 0.86 & 0.94 & 0.74 & 0.86 & 0.79 & 0.89 & 0.79 & 0.89 \\ \hline
DS3 & 0.86 & 0.94 & 0.89 & 1.00 & 0.86 & 0.89 & 0.86 & 0.94 & 0.86 & 0.94 \\ \hline
DS4 & 0.79 & 0.98 & 0.86 & 0.92 & 0.70 & 0.98 & 0.76 & 0.94 & 0.76 & 0.94 \\ \hline
\end{tabular}
\end{table*}

The classification results reported in Table \ref{tab:results} help to evaluate the interestingness of the selected patterns. Indeed, this will demonstrate if the unsubstituted patterns were arbitrarily selected or they are really representative. Table \ref{tab:results} shows that the classification accuracy significantly increases with all datasets. We notice a huge leap in accuracy especially with DS1 and DS4 with a gain of more than 17\% and reaching almost full accuracy with DS4. To better understand the accuracy results, we also reports the average precision, recall, F-measure and AUC values for all cases. We notice an enhancement of performance with all the mentioned quality metrics. This supports the reliability of our selection approach.
\subsection{Results Using Other Substitution Matrices}
Besides Blosum62, biologists also defined other substitution matrices describing the likelihood that two amino acid types would mutate to each other in evolutionary time. We want to study the effect of using other substitution matrices on the experimental results. Though, we perform the same experiments following the same protocol and settings but using two other substitution matrices, namely $Blosum80$ and $Pam250$. We compare the obtained results in terms of number of subgraphs and classification accuracy with those obtained using the whole set of frequent subgraphs and those using subgraphs selected by \textsc{UnSubPatt} with Blosum62. The results are reported in Table \ref{tab:results_matrices}. A high selection rate accompanied with a clear enhancement of the classification accuracy is noticed using \textsc{UnSubPatt} with all the substitution matrices compared to the results using the whole set of frequent subgraphs. It is clearly noticed that even using different substitution matrices, \textsc{UnSubPatt} shows relatively similar behavior and is able to select a small yet relevant subset of patterns. It is also worth mentioning that for all the datasets, the best classification accuracy is obtained using Blosum62 and the best selection rate is achieved using Pam250. This is simply due to how distant proteins within the same dataset are, since each substitution matrix was constructed to implicitly express a particular theory of evolution. Though, choosing the appropriate substitution matrix can influence the outcome of the analysis.

\begin{table*}[!t]
	\centering
	\caption{Number of subgraphs (\#SG) and accuracy (ACC) of the classification of each dataset using NB coupled with frequent subgraphs (FSg) then representative unsubstituted subgraphs using Blosum80 (UnSubPatt\_Blosum80) and Pam250 (UnSubPatt\_Pam250)\label{tab:results_matrices}}

\begin{tabular}{|c|c|c|c|c|c|c|c|c|}\hline
\multirow{2}{*}{\textbf{Dataset}} & \multicolumn{2}{|c}{\textbf{FSg}} & \multicolumn{2}{|c}{\textbf{UnSubPatt\_Blosum62}} & \multicolumn{2}{|c}{\textbf{UnSubPatt\_Blosum80}} & \multicolumn{2}{|c|}{\textbf{UnSubPatt\_Pam250}} \\
\cline{2-9}
 & \textbf{\#SG} & \textbf{Accuracy} & \textbf{\#SG} & \textbf{Accuracy} & \textbf{\#SG} & \textbf{Accuracy} & \textbf{\#SG} & \textbf{Accuracy}\\ \hline

DS1 & 799094 & 0.62 & 7291 & 0.78 & 7328 & 0.67 & 6137 & 0.68 \\ \hline
DS2 & 258371 & 0.80 & 15898 & 0.90 & 15930 & 0.87 & 15293 & 0.87\\ \hline
DS3 & 114793 & 0.86 & 14713 & 0.94 & 14792 & 0.91 & 14363 & 0.93\\ \hline
DS4 & 1073393 & 0.79 & 9958 & 0.98 & 10417 & 0.90 & 9148 & 0.90\\ \hline
\end{tabular}
\end{table*}

\subsection{Impact of Substitution Threshold}
In our experiments, we used a substitution threshold (of 30\%) to select the unsubstituted patterns from the set of discovered frequent subgraphs. In this section, we study the impact of variation of the substitution threshold on both the number of selected subgraphs and the classification results. To do so, we perform the same experiments while varying the substitution threshold from 0\% to 90\% with a step-size of 10. In order to check if the enhancements of the obtained results are due to our selected features or to the classifier, we use two other well-known classifiers namely the support vector machine (SVM) and decision tree (C4.5) besides the na\"ive bayes (NB) classifier. We use the same protocol and settings of the previous experiments.
Figure \ref{selection_rate} presents the selection rate for different substitution thresholds and Figures \ref{fig:NB}, \ref{fig:SVM} and \ref{fig:C4.5} illustrate the classification accuracy obtained respectively using NB, SVM and C4.5 with each dataset. The classification accuracy of the initial set of frequent subgraphs (gSpan, the line in red) is considered as a standard value for comparison. Thus, the accuracy values of \textsc{UnSubPatt} (in blue) that are above the line of the standard value are considered as gains, and those under the standard value are considered as losses.

In Figure \ref{selection_rate}, we notice that \textsc{UnSubPatt} reduces considerably the number of subgraphs especially with lower substitution thresholds. In fact, the number of unsubstituted patterns does not exceed 30\% for all substitution thresholds below 70\% and even reaches less then 1\% in some cases. This important reduction in the number of patterns comes with a notable enhancement of the classification accuracies. This fact is illustrated in Figures \ref{fig:NB}, \ref{fig:SVM} and \ref{fig:C4.5} which show that the unsubstituted patterns allow better classification performance compared to the original set of frequent subgraphs. \textsc{UnSubPatt} scores very well with the three used classifiers and even reaches full accuracy in some cases. This confirms our assumptions and shows that our selection is reliable and contributes to the enhancement of the accuracy. However, we believe that NB allows the most reliable evaluation because it performs a classification based on a global and conditional evaluation of features, unlike SVM which performs itself another attribute selection to select the support vectors and unlike C4.5 which performs an attribute by attribute evaluation.

In the case of proteins, a substitution threshold of 0\% enables selecting subgraphs based only on their structure. Precisely, \textsc{UnSubPatt} will select only one pattern from each group of subgraphs that are structurally isomorphic. Based on the experimental results, we believe that using these patterns is enough for a structural classification task since it allows a fast selection, selects a very small number of subgraphs and performs very well on classification.

\begin{figure}
	\centering
	\includegraphics[width=0.45\textwidth]{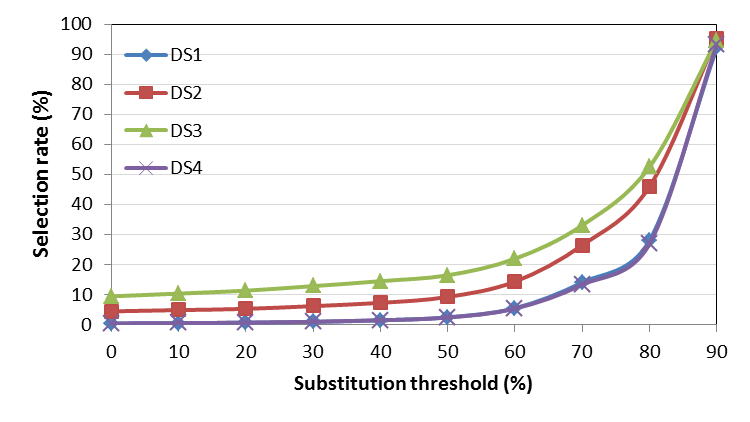}
	   \caption{Rate of unsubstituted patterns from $\Omega$ depending on the substitution threshold ($\tau$).}
	   \label{selection_rate}
\end{figure}

\begin{figure}
	\centering
	\includegraphics[width=0.45\textwidth]{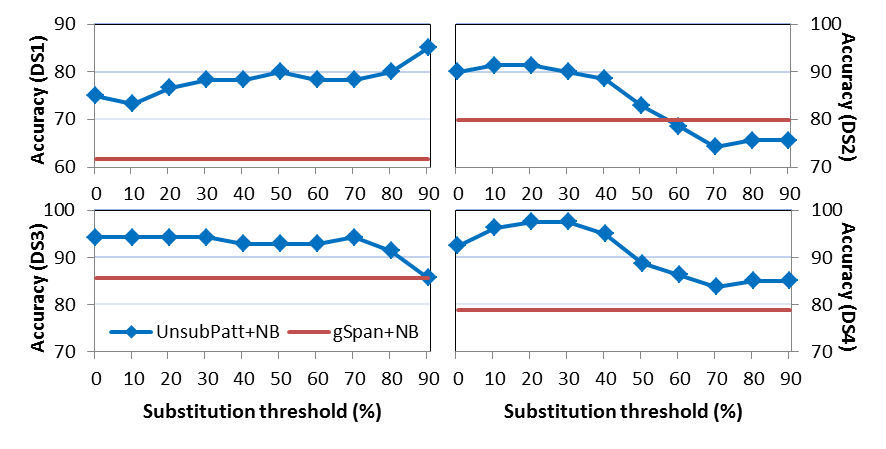}
		\caption{Classification accuracy by NB.}\label{fig:NB}
\end{figure}

\begin{figure}
	\centering
	\includegraphics[width=0.45\textwidth]{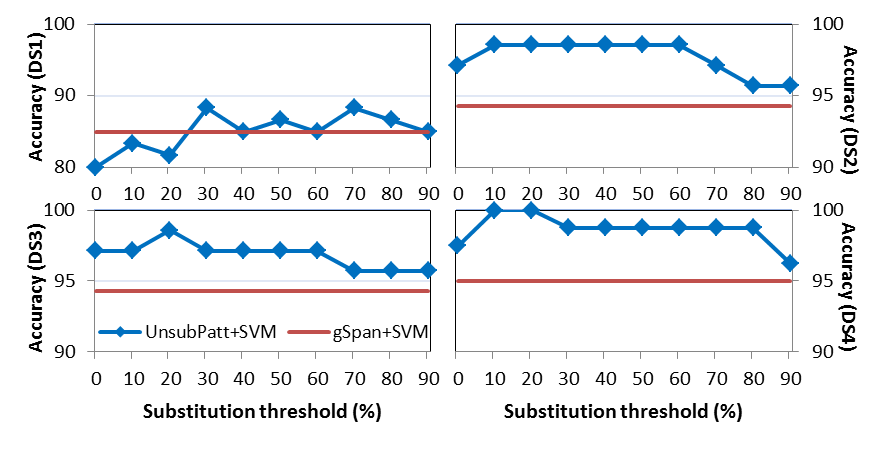}
	   \caption{Classification accuracy by SVM.}\label{fig:SVM}
\end{figure}

\begin{figure}
	\centering
	\includegraphics[width=0.45\textwidth]{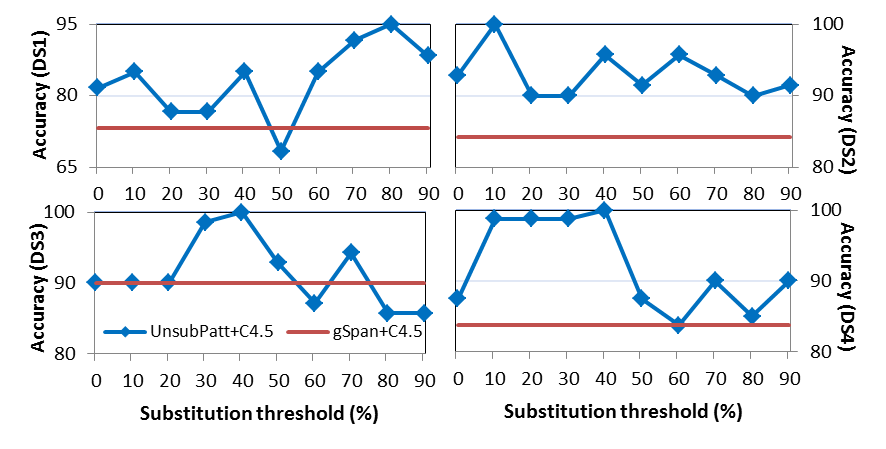}
	  \caption{Classification accuracy by C4.5.} \label{fig:C4.5}
\end{figure}

\subsection{Size-based Distribution of Patterns}
In this section, we study the distribution of patterns based on their size (number of edges). We try to check which sizes of patterns are more concerned by the selection.  The Figures \ref{fig:pattern_distribution_ds1} and \ref{fig:pattern_distribution_ds2} draw the distribution of patterns for the original set of frequent subgraphs and for the final set of representative unsubstituted subgraphs with all the substitution thresholds using Blosum62. The downward tendency of \textsc{UnSubPatt} using lower substitution thresholds and with respect to the original set of frequent subgraph is very clear. In fact, \textsc{UnSubPatt} leans towards cutting off the peaks and flattening the curves with lower substitution thresholds. Another interesting observation is that the curves are flattened in the regions of small patterns as well as in those of big and dense patterns. This demonstrates the effectiveness of \textsc{UnSubPatt} with both small and big patterns. 

\begin{figure}
\centering
	\includegraphics[width=0.45\textwidth]{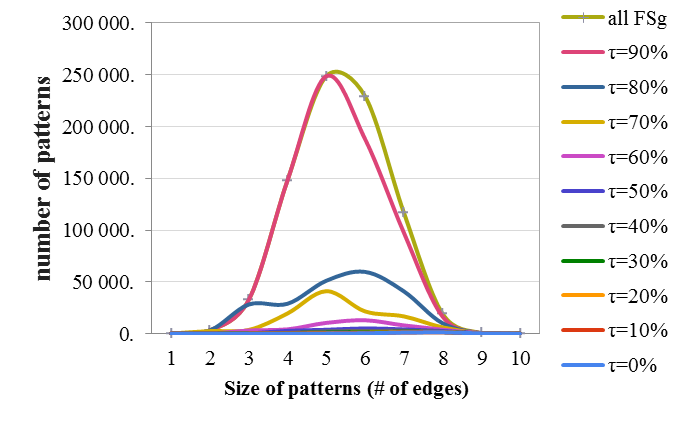}
	  \caption{Distribution of patterns of DS1 for all the frequent subgraphs and for the representative unsubstituted ones with all the substitution thresholds} \label{fig:pattern_distribution_ds1}
\end{figure}

\begin{figure}
\centering
	\includegraphics[width=0.45\textwidth]{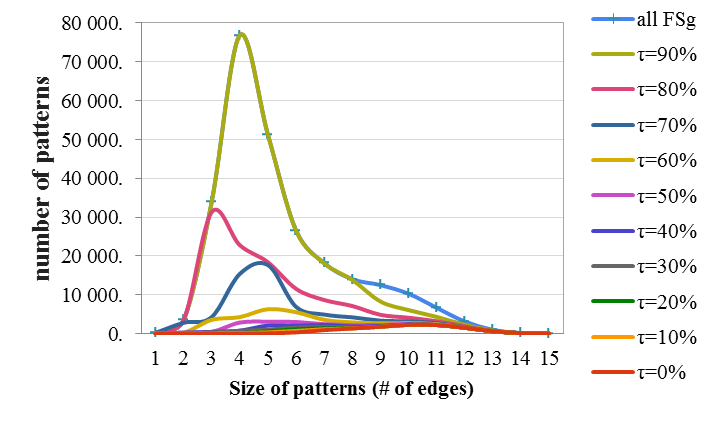}
	  \caption{Distribution of patterns of DS2 for all the frequent subgraphs and for the representative unsubstituted ones with all the substitution thresholds} \label{fig:pattern_distribution_ds2}
\end{figure}

\subsection{Comparison with other approaches}

To objectively evaluate our approach, we compare it with current trends in subgraph selection. In Figure \ref{fig:accuracy_comparison}, we report the classification accuracy using the representative unsubstituted patterns of \textsc{UnSubPatt} besides those using patterns of other new subgraph selection approaches from the literature namely LEAP\cite{Yan_2008}, gPLS\cite{Saigo_2008}, COM\cite{Jin_2009} and LPGBCMP\cite{Fei_2010} (reported and explained in the introductory section).

For \textsc{UnSubPatt}, we report the results obtained using the substitution matrix Blosum62, a minimum substitution threshold $\tau = 30\%$ and SVM for classification. For LEAP+SVM, LEAP is used iteratively to discover a set of discriminative subgraphs with a leap length=0.1. The discovered subgraphs are consider as features to train SVM with a 5-fold cross validation. COM is used with $t_p=30\%$ and $t_n=0\%$. For gPLS, the frequency threshold is set to 30\% and the best accuracies are reported for all the datasets among all the parameters combinations from m = {2, 4, 8, 16} and k = {2, 4, 8, 16}, where m is the number of iterations and k is the number of patterns obtained per search. For LPGBCMP, threshold values of $max_var=1$ and $\delta=0.25$ were respectively used for feature consistency map building and for overlapping. The obtained results are reported in the Figure \ref{fig:accuracy_comparison}.

\begin{figure}
\centering
	\includegraphics[width=0.45\textwidth]{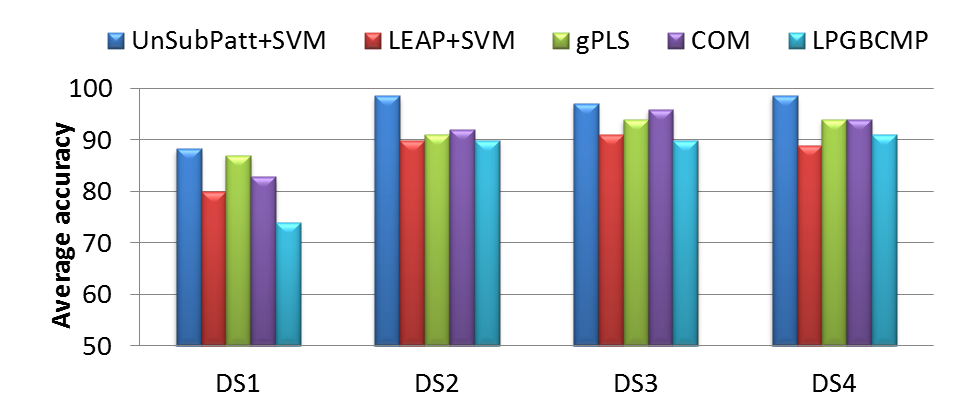}
	  \caption{Classification accuracy comparison with other pattern selection approaches.} \label{fig:accuracy_comparison}
\end{figure}

The classification results displayed in Figure \ref{fig:accuracy_comparison} show that \textsc{UnSubPatt} allows a better classification than all the other pattern selection methods in the four cases. Considering only these results does not allow to confirm that \textsc{UnSubPatt} would always outperform the considered methods. However, this proves that \textsc{UnSubPatt} represents a very competitive and promising approach. It is also worth noting that these approaches are supervised and dedicated to classification unlike \textsc{UnSubPatt} which is an unsupervised approach. This allows it to be used in classification as well as in other mining tasks such as clustering and indexing.

\subsection{Runtime Analysis}
To study the variation of \textsc{UnSubPatt}'s runtime with larger amounts of data, we use different sets of frequent patterns from 10000 to 100000 with step-size of 10000. In Table \ref{tab:runtime}, we report the runtime results for the pattern sets using three substitution thresholds.
\begin{table}[!t]
	\centering
	\caption{Runtime analysis of \textsc{UnSubPatt} with different substitution thresholds}
	\label{tab:runtime}
\begin{tabular}{|c|c|c|c|}
\hline 
\textbf{Number of}&\multicolumn{3}{|c|}{\textbf{Substitution thresholds}} \\ 
\cline{2-4} 
\textbf{patterns} & \textbf{$\tau$ = 10\%} & \textbf{$\tau$ = 30\%} & \textbf{$\tau$ = 50\%} \\ 
\hline 
10000 & 4s & 4s & 4s \\ 
\hline 
20000 & 8s & 8s & 10s \\ 
\hline 
30000 & 13s & 13s & 17s \\ 
\hline 
40000 & 18s & 18s & 25s \\ 
\hline 
50000 & 23s & 23s & 33s \\ 
\hline 
60000 & 28s & 28s & 41s \\ 
\hline 
70000 & 35s & 35s & 52s \\ 
\hline 
80000 & 40s & 42s & 66s \\ 
\hline 
90000 & 46s & 49s & 80s \\ 
\hline 
100000 & 53s & 57s & 136s \\ 
\hline 
\end{tabular} 
\end{table} 

Even though the complexity of the problem due to the combinatorial test of substitution between subgraphs, our algorithm is scalable with higher amounts of data. With increasing number of patterns, the runtime is still reasonable. The use of different substitution thresholds slightly affected the runtime of \textsc{UnSubPatt}, since the number of selected patterns is comparable for all thresholds.

A possible way to make \textsc{UnSubPatt} runs faster is parallelization. In fact, \textsc{UnSubPatt} can be easily parallelized, since it tests separately the substitution among each group of subgraphs having the same size and order. Hence, these groups can be distributed and treated separately in different processes.

\section{Conclusion}
In this paper, we proposed a novel selection approach for mining a representative summary of the set of frequent subgraphs. Unlike existent methods that are based on the relations between patterns in the transaction space, our approach considers the distance between patterns in the pattern space. The proposed approach exploits a specific domain knowledge, in the form of a substitution matrix, to select a subset of representative unsubstituted patterns from a given set of frequent subgraphs. It also allows to reduce considerably the size of the initial set of subgraphs to obtain an interesting and representative one enabling easier and more efficient further explorations. It is also worth mentioning that our approach can be used on graphs as well as on sequences and is not limited to classification tasks, but can help in other subgraph-based analysis such as indexing, clustering, visual inspection, etc.

Since the proposed approach is a filter approach, a promising future direction could be to find a way to integrate the selection within the extraction process in order to directly mine the representative patterns from data. Moreover, we intend to use our approach in other classification contexts and in other mining applications.

\section{Acknowledgement}
This work is supported by a PhD grant from the French Ministry of Higher Education to the first author.

%
\bibliographystyle{abbrv}
%
%


\end{document}